\documentclass[review]{elsarticle}

\usepackage{lineno,hyperref}
\modulolinenumbers[5]

\usepackage[T1]{fontenc}
\usepackage[latin9]{inputenc}
\usepackage{color}
\usepackage{array}
\usepackage{float}
\usepackage{mathrsfs}
\usepackage{mathtools}
\usepackage{amsthm}
\usepackage{amstext}
\usepackage{amssymb}
\usepackage{stmaryrd}
\usepackage{graphicx}
\usepackage{pgfplots}

\makeatletter

\theoremstyle{plain}
\newtheorem{theorem}{\protect\theoremname}
\theoremstyle{plain}

\theoremstyle{plain}
\newtheorem{corollary}{\protect\corollaryname}
\theoremstyle{plain}
\newtheorem{lemma}{\protect\lemmaname}
\theoremstyle{definition}
\newtheorem{example}{\protect\examplename}
\theoremstyle{definition}

\theoremstyle{definition}

\AtBeginDocument{
	
}

\makeatother

  \providecommand{\corollaryname}{Corollary}
  \providecommand{\examplename}{Example}
  \providecommand{\lemmaname}{Lemma}
  \providecommand{\propositionname}{Proposition}
  \providecommand{\theoremname}{Theorem}
  \providecommand{\definitionname}{Definition}
  \providecommand{\remarkname}{Remark}

\newcommand{\Bgk}[1]{{\Bigg( {#1} \Bigg) }}

\begin{document}

\begin{frontmatter}

\title{Complete weight enumerators of a class of linear codes with two or three weights \tnoteref{mytitlenote}}
\tnotetext[mytitlenote]{The work is partially supported by the National Natural Science Foundation of China (11701317,61772015,61472457,11571380), China Postdoctoral Science Foundation Funded Project (2017M611801) and Jiangsu Planned Projects for Postdoctoral Research Funds (1701104C). This work is also partially supported by Guangzhou Science and Technology Program (201607010144).}

\author[mymainaddress]{Shudi Yang\corref{mycorrespondingauthor}}
\cortext[mycorrespondingauthor]{Corresponding author}
\ead{yangshudi7902@126.com}
\author[mymainaddress]{Xiangli Kong}
\ead{kongxiangli@126.com}

\address[mymainaddress]{School of Mathematical
	Sciences, Qufu Normal University, Shandong 273165, P.R.China}

\begin{abstract}
We construct a class of linear codes by choosing a proper defining set and determine their complete weight enumerators and weight enumerators. The results show that they have at most three weights and they are suitable for applications in secret sharing schemes. This is an extension of the results raised by Wang $ \emph{et al.} $ (2017).
\end{abstract}


\end{frontmatter}


\section{Introduction}\label{sec:intro}

Throughout this paper, let $p$ be an odd prime and $q=p^e$ for a positive integer $e$. Denote by $\mathbb{F}_q$ a
finite field with $q$ elements. An $[n, \kappa, \delta]$ linear code
$C$ over $\mathbb{F}_p$ is a $\kappa$-dimensional subspace of
$\mathbb{F}_p^n$ with minimum distance $\delta$
(see~\cite{macwilliams1977theory}).
Let $A_i$ denote the number of codewords with Hamming weight
$i$ in a linear code $C$ of length $n$. The weight enumerator of $C$ is defined by
$1+A_1z+A_2z^2+\cdots+A_nz^n $. A code $ C $ is called a $ t $-weight code if there are $ t $ nonzero $ A_i  $ for $ 1\leq i \leq n $.

The complete weight enumerator of a code $C$ over $\mathbb{F}_p$ enumerates the codewords according to the number of symbols of each kind contained in each codeword. Denote elements of the field by $\mathbb{F}_p=\{w_0=0,w_1,\cdots,w_{p-1}\}$.
Also, let $\mathbb{F}_p^*$ denote $\mathbb{F}_p\backslash\{0\}$.
For a codeword $\mathsf{c}=(c_0,c_1,\cdots,c_{n-1})\in \mathbb{F}_p^n$, let $w[\mathsf{c}]$ be the
complete weight enumerator of $\mathsf{c}$, which is defined as
$$w[\mathsf{c}]=w_0^{k_0}w_1^{k_1}\cdots w_{p-1}^{k_{p-1}},$$
where $k_j$ is the number of components of $\mathsf{c}$ equal to $w_j$, $\sum_{j=0}^{p-1}k_j=n$.
The complete weight enumerator of the code $C$ is then
$$\mathrm{CWE}(C)=\sum_{\mathsf{c}\in C}w[\mathsf{c}].$$

The weight enumerators of linear codes have been well studied in literature, such as \cite{dinh2015recent,feng2008weight,vega2012weight,YangYao2017,zheng2015weightseveral,Zhou2013fiveweight} and references therein. The information of the complete weight enumerators of linear codes is of vital use because they can show the frequency of each symbol appearing in each codeword. Furthermore the complete weight enumerator has close relation to the deception probabilities of certain authentication codes \cite{ding2007generic}, and is used to compute the Walsh transform of monomial and quadratic bent functions over finite fields \cite{helleseth2006}.
More researches can be found in \cite{Bae2015complete,Blake1991,Ding2005auth,kith1989complete,kuzmin1999complete}.

We introduce a generic construction of linear codes developed in \cite{ding2015twodesign,dingkelan2014binary,ding2015twothree}. Set $D=\{d_1,d_2,\cdots,d_n\}\subseteq \mathbb{F}_{q}$, where $q=p^e$. Denote by $\mathrm{Tr}$ the absolute trace function from $ \mathbb{F}_{q} $ to $ \mathbb{F}_{p} $. A linear code of length $ n=\#D $ is defined by
\begin{equation}\label{def:CD}
C_{D}=\{(\mathrm{Tr}(bd_1),\mathrm{Tr}(bd_2),\cdots,\mathrm{Tr}(bd_n)):
b\in \mathbb{F}_{q}\}.
\end{equation}
The set $D$ is called the defining set of $C_{D}$.
This construction technique is general and has received a good deal of attention, see \cite{AhnKaLi2016completegenelize,Heng2016cwecylic,LiYang2015cwe,li2015complete,WangQiuyan2015complet,Yang2017constru,Yang2016complete} for more details.

Let $ d=\gcd(k,e) $ be the greatest common divisor of positive integers $ k $ and $ e $. Suppose that $ e/d $ is even with $ e=2m $ and $ m >1 $. Motivated by the above construction and the idea of \cite{WangQiuyan2015complet}, we investigate a class of linear codes with defining set
\begin{align*} 
D_c=\{ x \in \mathbb{F}_{q}^* :\mathrm{Tr}(ax^{p^k+1}) =c\} ,
\end{align*}
where $ a \in \mathbb{F}_{q}^* $ and $ c\in \mathbb{F}_p $.
We will extend the results presented by Wang $ \emph{et al.} $ \cite{WangQiuyan2015complet} who studied the case where $ a=1 $.

The remainder of this paper is organized as follows. In Section
\ref{sec:main results}, we describe the main results of this paper, additionally we give some examples. In Section~\ref{sec:pf}, we briefly recalls some definitions and results on exponential sums, then proves the main results.
In Section~\ref{sec:conclusion}, we make a conclusion.

\section{Main results}\label{sec:main results}

In this section, we only introduce the complete weight enumerator and weight enumerator of $C_{D_c}$ with defining set
$ D_c $. The main results of this paper are presented below, whose proofs will be given in Section~\ref{sec:pf}. Denote $ s=m/d $.

\begin{theorem}\label{thm:D0}
	If $ a^{(q-1)/(p^d+1)} \neq (-1)^s $, then the code $ C_{D_0} $ of \eqref{def:CD} is a $ [p^{e-1}+	(-1)^s (p-1)p^{m-1} -1, e] $ linear code with weight enumerator
	\begin{align*}
	1+ (p^{e-1}+	(-1)^s (p-1)p^{m-1} -1) z^{(p-1) p^{e-2}}
	+ (p-1)(p^{e-1}-	(-1)^s p^{m-1} ) z^{(p-1) (p^{e-2}+(-1)^s p^{m-1})},
	\end{align*}
	and its complete weight enumerator is
	\begin{align*}
	& w_0^{p^{e-1}+	(-1)^s (p-1)p^{m-1} -1}
	+ (p^{e-1}+	(-1)^s (p-1)p^{m-1} -1) w_0^{ p^{e-2} +	(-1)^s (p-1)p^{m-1}  -1}  \prod_{\rho \in\mathbb{F}_{p}^* } w_{\rho}^{  p^{e-2} }\\
	&+ (p-1)(p^{e-1}-	(-1)^s p^{m-1} ) w_0^{ p^{e-2} -1} \prod_{\rho \in\mathbb{F}_{p}^* }w_{\rho}^{  p^{e-2}+(-1)^s p^{m-1}}.
	\end{align*}
	
\end{theorem}

\begin{theorem}\label{thm:D1}
	If $ m>d+1 $ and $ a^{(q-1)/(p^d+1)} = (-1)^s $, then the code $ C_{D_0} $ of \eqref{def:CD} is a $ [p^{e-1}-	(-1)^s (p-1)p^{m+d-1} -1, e] $ linear code with weight enumerator
	\begin{align*}
	1&+ (p^e-p^{e-2d}) z^{(p-1) (p^{e-2}-	(-1)^s (p-1)p^{m+d-2}) }
	+ (p^{e-2d-1}-	(-1)^s (p-1)p^{m-d-1}-1 ) z^{(p-1) p^{e-2} }\\
	&+(p-1)(p^{e-2d-1}+	(-1)^s p^{m-d-1})z^{(p-1)(p^{e-2}-	(-1)^s  p^{m+d-1})},
	\end{align*}
	and its complete weight enumerator is
	\begin{align*}
	& w_0^{p^{e-1}-	(-1)^s (p-1)p^{m+d-1} -1}
	+ (p^{e}-p^{e-2d} ) w_0^{ p^{e-2} -	(-1)^s (p-1)p^{m+d-2}  -1}
	\prod_{\rho \in\mathbb{F}_{p}^* } w_{\rho}^{  p^{e-2}-	(-1)^s (p-1)p^{m+d-2} }\\
	&+ (p^{e-2d-1}-	(-1)^s (p-1)p^{m-d-1}-1 )
	w_0^{ p^{e-2}-	(-1)^s (p-1)p^{m+d-1} -1}
	\prod_{\rho \in\mathbb{F}_{p}^* }w_{\rho}^{  p^{e-2} }\\
	&+(p-1)(p^{e-2d-1}+	(-1)^s p^{m-d-1})
	w_0^{ p^{e-2} -1}
	\prod_{\rho \in\mathbb{F}_{p}^* }w_{\rho}^{p^{e-2}-	(-1)^s  p^{m+d-1}} .
	\end{align*}
	
\end{theorem}

When $ m=d+1 $ or $ m=d $, the corresponding results are described below.
\begin{corollary}
	If $ m=2 $, $ d=1 $ and $ a^{(q-1)/(p^d+1)} = 1 $, the set $ C_{D_0} $ of \eqref{def:CD} is a multi set with each codewords appear $ p^2 $ times. It is a $ [p^2-	 1, 2,(p-1) p ] $ linear code with only one nonzero weight $ (p-1) p $.
\end{corollary}

\begin{corollary}
	If $ m=d $ and $ a^{(q-1)/(p^d+1)} = -1 $, then the code $ C_{D_0} $ of \eqref{def:CD} is a $ [p^e-	 1, e,(p-1)  p^{e-1}] $ linear code with only one nonzero weight
	$ (p-1)  p^{e-1}   $.
\end{corollary}

\begin{theorem}\label{thm:D2}
	Let $ g $ be a generator of $ \mathbb{F}_p^* $ and $ c\in \mathbb{F}_p^* $. If $ a^{(q-1)/(p^d+1)} \neq (-1)^s $, then the code $ C_{D_c} $ of \eqref{def:CD} is a $ [p^{e-1}-	(-1)^s p^{m-1}, e] $ linear code with weight enumerator
	\begin{align*}
	1+ \big(\frac{p+1}{2}p^{e-1}+ \frac{p-1}{2}	(-1)^s p^{m-1} -1\big) z^{(p-1) p^{e-2}}
	+ \frac{p-1}{2}(p^{e-1}-	(-1)^s p^{m-1} ) z^{(p-1) p^{e-2}-2(-1)^s p^{m-1}},
	\end{align*}
	and its complete weight enumerator is
	\begin{align*}
	& w_0^{p^{e-1}-	(-1)^s p^{m-1} }
	+ (p^{e-1}+	(-1)^s (p-1)p^{m-1} -1) w_0^{ p^{e-2} -	(-1)^s p^{m-1} }  \prod_{\rho \in\mathbb{F}_{p}^* } w_{\rho}^{  p^{e-2} }\\
	&+ (p^{e-1}-	(-1)^s p^{m-1} ) \sum_{\beta=1}^{ \frac{p-1}{2}}
	w_0^{ p^{e-2}+(-1)^s\eta(-1) p^{m-1} }
	w_{2g^{\beta}}^{ p^{e-2}} w_{p-2g^{\beta}}^{ p^{e-2}}\prod_{\substack{\rho \in\mathbb{F}_{p}^*\\ \rho \neq 0,\pm 2g^{\beta} } }w_{\rho}^{  p^{e-2}+(-1)^s p^{m-1} \eta (\rho^2-4g^{2\beta})}\\
	&+ (p^{e-1}-	(-1)^s p^{m-1} ) \sum_{\beta=1}^{ \frac{p-1}{2}}
	w_0^{ p^{e-2}-(-1)^s\eta(-1) p^{m-1} }\prod_{\rho \in\mathbb{F}_{p}^* }w_{\rho}^{  p^{e-2}+(-1)^s p^{m-1} \eta (\rho^2-4g^{2\beta+1})},
	\end{align*}
	where $ \eta $ is the quadratic character over $ \mathbb{F}_p^* $.
\end{theorem}

\begin{theorem}\label{thm:D3}
	Let $ g $ be a generator of $ \mathbb{F}_p^* $ and $ c\in \mathbb{F}_p^* $. If  $ a^{(q-1)/(p^d+1)} = (-1)^s $, then the code $ C_{D_c} $ of \eqref{def:CD} is a $ [p^{e-1}+	(-1)^s p^{m+d-1}, e] $ linear code with weight enumerator
	\begin{align*}
	1&+ (p^e-p^{e-2d}) z^{(p-1) (p^{e-2}+	(-1)^s  p^{m+d-2}) }
	+ \big(\frac{p+1}{2} p^{e-2d-1}-\frac{p-1}{2}	(-1)^s p^{m-d-1}-1 \big) z^{(p-1) p^{e-2} }\\
	&+\frac{p-1}{2}(p^{e-2d-1}+	(-1)^s p^{m-d-1})z^{(p-1)p^{e-2}+2	(-1)^s  p^{m+d-1}},
	\end{align*}
	and its complete weight enumerator is
	\begin{align*}
	& w_0^{p^{e-1}+	(-1)^s p^{m+d-1} }
	+ (p^{e}-p^{e-2d} ) w_0^{ p^{e-2} +	(-1)^s  p^{m+d-2}   }
	\prod_{\rho \in\mathbb{F}_{p}^* } w_{\rho}^{  p^{e-2}+	(-1)^s p^{m+d-2} }\\
	&+ (p^{e-2d-1}-	(-1)^s (p-1)p^{m-d-1}-1 )
	w_0^{ p^{e-2}+	(-1)^s  p^{m+d-1}  }
	\prod_{\rho \in\mathbb{F}_{p}^* }w_{\rho}^{  p^{e-2} }\\
	&+ (p^{e-2d-1}+	(-1)^s p^{m-d-1})
	\sum_{\beta=1}^{ \frac{p-1}{2}}
	w_0^{ p^{e-2}-(-1)^s\eta(-1) p^{m+d-1} }
	w_{2g^{\beta}}^{ p^{e-2}} w_{p-2g^{\beta}}^{ p^{e-2}}\prod_{\substack{\rho \in\mathbb{F}_{p}^* \\\rho \neq 0,\pm 2g^{\beta} }  }w_{\rho}^{  p^{e-2}-(-1)^s p^{m+d-1} \eta (\rho^2-4g^{2\beta})} \\
	&+ (p^{e-2d-1}+	(-1)^s p^{m-d-1})
	\sum_{\beta=1}^{ \frac{p-1}{2}}
	w_0^{ p^{e-2}+(-1)^s\eta(-1) p^{m+d-1} }  \prod_{\rho \in\mathbb{F}_{p}^*  }w_{\rho}^{  p^{e-2}-(-1)^s p^{m+d-1} \eta (\rho^2-4g^{2\beta+1})}  ,
	\end{align*}
	where $ \eta $ is the quadratic character over $ \mathbb{F}_p^* $.
\end{theorem}

Some concrete examples are provided to illustrate our results.
\begin{example}
	Let $(p,m,k)=(3,3,1)$, $ \mathbb{F}_{3^6}^*=\langle\theta\rangle $ and $ a=\theta^2 $. Then $ d=1 $, $ s=3 $ and  $ a^{(q-1)/(p^d+1)} =-1  $. If $ c=0 $, by Theorem \ref{thm:D1}, the corresponding code $C_{D_0}$ is a $[296, 6, 162]$ linear code. Its weight enumerator is
	$ 1+32 z^{162} + 648 z^{198} + 48 z^{216} $, and its complete weight enumerator is
	\begin{align*}
	w_0^{296}+32 w_0^{134}( w_1 w_2)^{81}   + 648 w_0^{98}( w_1 w_2)^{99}   + 48 w_0^{80}( w_1 w_2 )^{108}  .
	\end{align*}	
	If $ c \neq 0 $, by Theorem \ref{thm:D3}, the corresponding code $C_{D_1}$ has parameters $[216, 6, 108]$   with weight enumerator
	$ 1+24 z^{108} + 648 z^{144} + 56 z^{162} $ and complete weight enumerator
	\begin{align*}
	w_0^{216}+24 w_0^{108}( w_1 w_2)^{54}   + 648  (w_0 w_1 w_2)^{72}   + 56 w_0^{54}( w_1 w_2 )^{81}  .
	\end{align*}
	These results coincide with the numerical computation by Magma.
	
\end{example}

\begin{example}
	Let $(p,m,k)=(5,2,1)$, $ \mathbb{F}_{5^4}^*=\langle\theta\rangle $ and $ a=\theta^3 $. Then $ d=1 $, $ s=2 $ and  $ a^{(q-1)/(p^d+1)} \neq 1  $. By Theorem \ref{thm:D0}, the code $C_{D_0}$ has parameters $[144, 4, 100]$ with weight enumerator
	$ 1+144 z^{100} + 480 z^{120}  $ and complete weight enumerator
	\begin{align*}
	w_0^{144}+144 w_0^{44}( w_1 w_2 w_3 w_4)^{25}   + 480   w_0^{24}( w_1 w_2 w_3 w_4)^{30}.
	\end{align*}	
     By Theorem \ref{thm:D2}, the code $C_{D_1}$ is a $[120, 4, 90]$ linear code. Its weight enumerator is
	$ 1+240 z^{90} + 384 z^{100} $, and its complete weight enumerator is
	\begin{align*}
	&w_0^{120}+120 w_0^{30}( w_1 w_4)^{25} ( w_2 w_3)^{20}
	+120 w_0^{30}( w_1 w_4)^{20} ( w_2 w_3)^{25}
	+120 w_0^{20}( w_1 w_4)^{30} ( w_2 w_3)^{20} \\
	&+120 w_0^{20}( w_1 w_4)^{20} ( w_2 w_3)^{30}
	+144 w_0^{20}( w_1 w_2 w_3 w_4)^{25} .
	\end{align*}
	These results coincide with the numerical computation by Magma.
\end{example}

\section{The proofs of the main results}\label{sec:pf}

\subsection{Auxiliary results}\label{subsec:math tool}
In order to prove the results proposed in Section~\ref{sec:main results}, we will use several results which are depicted and proved in the sequel. We start with group characters and exponential sums.
For each $b\in\mathbb{F}_q$, an additive character $\chi_b$ of $\mathbb{F}_q $ is defined by
$ \chi_b(x)=\zeta_p^{\text{Tr}(bx)}$ for all $x\in\mathbb{F}_q $, where
$\zeta_p=\exp\left(\frac{2\pi \sqrt{-1} }{p}\right)$
and $\text{Tr}$ is the simplification of the trace function $ \text{Tr}^e_1  $
from $ \mathbb{F}_q $ to $ \mathbb{F}_p  $. For $b=1$, $\chi_1$ is called the canonical additive character of $\mathbb{F}_q$.

Let $\eta_e$ denote the quadratic character of $\mathbb{F}_q^*$ and is extended by $ \eta_e(0)=0 $. The quadratic Gauss sum $G(\eta_e,\chi_1)$ is defined by
\begin{align*}
G(\eta_e, \chi_1 )=\sum_{x\in\mathbb{F}_q^*}\eta_e(x)\chi_1(x).
\end{align*}
We denote $G_e=G(\eta_e, \chi_1 )$ and $G=G(\eta, \chi_1')$,
where $ \eta$ and $\chi_1'$ are the quadratic character and canonical additive character of $\mathbb{F}_{p}$, respectively.
Moreover, it is well known that $G_e=(-1)^{e-1}\sqrt{p^*}^e$ and $G =\sqrt{p^*}$, where $p^*= \eta(-1)p $. See~\cite{ding2015twothree,lidl1983finite} for more information.

The following lemmas will be of special use in the sequel.
\begin{lemma}[Theorem 5.33, \cite{lidl1983finite}]\label{lm:expo sum}
	Let $ q =p^e$ be odd and $f(x)=a_2x^2+a_1x+a_0\in \mathbb{F}_{q}[x]$ with
	$a_2\neq0$.	Then
	\begin{equation*}
	\sum_{x\in
		\mathbb{F}_{q}}\zeta_p^{\mathrm{Tr}(f(x))}=\zeta_p^{\mathrm{Tr}(a_0-a_1^2(4a_2)^{-1})}\eta_e(a_2)G_e,
	\end{equation*}
	where $ \eta_e $ is the quadratic character of $ \mathbb{F}_{q} $.
\end{lemma}

\begin{lemma}
	[Theorem 5.48, \cite{lidl1983finite}]\label{lm:eta(f)}
	With the notation of Lemma \ref{lm:expo sum}, we have
	\begin{align*}
	\sum_{x\in
		\mathbb{F}_{q}}\eta_e (f(x))= \left\{\begin{array}{lll}
	-\eta_e(a_2) & \textup{ if } a_1^2-4a_0a_2 \neq 0  ,\\
	(q-1)\eta_e(a_2) & \textup{ if }  a_1^2-4a_0a_2 = 0. 	
	\end{array}
	\right.
	\end{align*}
\end{lemma}

For $ \alpha,\beta \in \mathbb{F}_q $ and any integer $ k $, the exponential sum  $ S(\alpha,\beta) $ is defined  by
\begin{align*}
S(\alpha,\beta)=\sum_{x\in \mathbb{F}_q } \chi_1(\alpha x^{p^k+1}+\beta x).
\end{align*}
We recall some results of  $ S(\alpha,\beta) $ for $ \alpha \neq 0 $ and $ q $ odd.
\begin{lemma}[Theorem 2, \cite{coulter1998explicit}]\label{lem:S(a0)}
	Let $ d=\gcd(k,e) $ and $ e/d $ be even with $ e=2m $. Then
	\begin{align*}
	S(\alpha,0)=\left\{\begin{array}{lll}
	(-1)^{s} p^m && \textup{ if } \alpha^{(q-1)/(p^d+1)} \neq (-1)^{s}  ,\\
	(-1)^{s+1}p^{m+d} && \textup{ if } \alpha^{(q-1)/(p^d+1)} = (-1)^{s}, 	
	\end{array}
	\right.
	\end{align*}
	where $ s=m/d $.
\end{lemma}
\begin{lemma}[Theorem 4.7, \cite{Coulter2002number}]\label{lem:S(ab)}
	Let $ \beta \neq 0 $ and $ e/d $ be even with $ e=2m $. Then $ S(\alpha,\beta)=0 $ unless the equation $ \alpha^{p^k} X^{p^{2k}} + \alpha X=-\beta^{p^k} $ is solvable. There are two possibilities.
	
	$ (i)  $ If $  \alpha^{(q-1)/(p^d+1)} \neq (-1)^{s} $, then for any choice of
	$ \beta \in \mathbb{F}_q $, the equation has a unique solution $ x_0 $ and
	\begin{align*}
	S(\alpha,\beta)=(-1)^{s}p^m \chi_1(-\alpha x_0^{p^k+1}).
	\end{align*}
	
	$ (ii)  $ If $  \alpha^{(q-1)/(p^d+1)} = (-1)^{s} $ and if the equation is solvable with some solution $ x_0 $ say, then
	\begin{align*}
	S(\alpha,\beta)=(-1)^{s+1}p^{m+d} \chi_1(-\alpha x_0^{p^k+1}).
	\end{align*}
	
\end{lemma}

\begin{lemma}[Theorem 4.1, \cite{coulter1998explicit}]\label{lem:fun}
	For	$  e = 2m $ the equation
	$
	\alpha^{p^k} X^{p^{2k} } +\alpha X=0
	$
	is solvable for  $ X \in \mathbb{F}_q^* $ if and only if $ e/d  $ is even and
	\begin{align*}
	\alpha^{(q-1)/(p^d+1)} = (-1)^{s}.
	\end{align*}
	In such cases there are $ p^{2d}- 1 $ non-zero solutions.
\end{lemma}

It follows that $   f_{\alpha} $ is a permutation polynomial
of $ \mathbb{F}_q $ with $ q = p^e $ if and only if $  e/d $ is odd or $ e/d $ is even with $ e = 2m $ and
$ \alpha^{(q-1)/(p^d+1)} \neq (-1)^{s} $.

\subsection{The proofs of Theorems in Section \ref{sec:main results}}\label{subsec:pf2}

In this subsection, we will give the proofs of our main results presented in Section \ref{sec:main results}.
Recall that $ q=p^e  $, $ d=\gcd(k,e) $, $ e/d $ is even with $ e=2m $. The code $ C_{D_c} $ with $ c\in \mathbb{F}_p $, is defined by
\begin{align*}
C_{D_c}=\{\mathsf{c}_b=(\mathrm{Tr}(bx))_{x\in D_c}: b\in \mathbb{F}_q \},
\end{align*}
where
$ D_c=\{ x \in \mathbb{F}_{q}^* :\mathrm{Tr}(ax^{p^k+1}) =c\} $,
with $ a \in \mathbb{F}_{q}^* $ and $ c\in \mathbb{F}_p $.
For convenience we define $ n_c=\#\{ x \in \mathbb{F}_{q}  :\mathrm{Tr}(ax^{p^k+1}) =c\} $. Then the length of the code is $ n_0 -1 $ for $ c=0 $ and otherwise $ n_c  $ for $ c\neq0 $.

\begin{lemma}\label{lem:length}
	Let $ c \in\mathbb{F}_{p} $. Denote $ s=m/d $.
	\begin{itemize}
		\item [1.] If $ c=0 $, then 	
		\begin{align*}
		n_0=\left\{\begin{array}{lll}
		p^{e-1} +	(-1)^s (p-1)p^{m-1}
		&& \textup{ if }  a^{(q-1)/(p^d+1)} \neq (-1)^s,\\
		p^{e-1}  	-(-1)^s (p-1)p^{m+d-1}
		&& \textup{ if }   a^{(q-1)/(p^d+1)} = (-1)^s.
		\end{array}
		\right.
		\end{align*}
		\item [2.] If $ c\neq 0 $, then 	
		\begin{align*}
		n_c=\left\{\begin{array}{lll}
		p^{e-1} -(-1)^s p^{m-1} \phantom{(p-1)}
		&& \textup{ if }   a^{(q-1)/(p^d+1)} \neq (-1)^s,\\
		p^{e-1} + (-1)^s p^{m+d-1}  \phantom{(p-1)}
		&& \textup{ if }   a^{(q-1)/(p^d+1)} = (-1)^s.
		\end{array}
		\right.
		\end{align*}
	\end{itemize}
	
\end{lemma}
\begin{proof}
	It follows that
	\begin{align*}
	n_c & = \frac{1}{p} \sum_{x \in \mathbb{F}_q  }
	\sum_{y \in \mathbb{F}_p } \zeta_p^{y \mathrm{Tr}(ax ^{p^k+1} ) -yc }\\
	& =  p^{e-1}+ p^{-1} \sum_{y \in \mathbb{F}_p^* } \zeta_p^{ -yc} \sum_{x \in \mathbb{F}_q  }	\zeta_p^{y \mathrm{Tr}(ax ^{p^k+1} )   }   \\
	& =  p^{e-1}+ p^{-1}
	\sum_{y \in \mathbb{F}_p^* } \zeta_p^{-yc} S(ay,0) .
	\end{align*}
	A straightforward calculation gives that  $ y^{(q-1)/(p^d+1)}=1 $ for $ y\in\mathbb{F}_{p}^* $. Then the desired conclusion follows from Lemma \ref{lem:S(a0)}.
\end{proof}

In order to investigate the weight enumerators of $ C_{D_c} $, we need to do some preparations. Observe that $b=0$ gives the zero codeword. Hence, we may assume that $b\neq 0$ in the rest of this subsection.
Let $  N_{ \rho} (b,c) $ denote the number of components  $  \mathrm{Tr}(bx ) $
of $ \mathsf{c}_b $  that are equal to  $ \rho $, where $  \rho\in \mathbb{F}_{p} $, $  b\in \mathbb{F}_{q}^* $ and $  c \in \mathbb{F}_{p} $. That is
\begin{align*}
N_{ \rho}(b,c)=\# \{ x \in \mathbb{F}_{q}:\mathrm{Tr}(a x ^{p^k+1} )=c \textup{ and }   \mathrm{Tr}(bx)=\rho\}.
\end{align*}
Then it is easy to obtain the Hamming weight of $ \mathsf{c}_b $, that is
\begin{align}\label{eq:N0}
wt(\mathsf{c}_b)=\sum_{\rho \in \mathbb{F}_{p}^*  } N_{ \rho}(b,c) = n_c-N_{ 0}(b,c).
\end{align}
So we only consider $ \rho \in \mathbb{F}_{p}^* $ and $ b\in \mathbb{F}_{q}^* $ in the sequel.

By the definition of $ N_{ \rho}(b,c) $, we have
\begin{align}\label{eq:Nab}
N_{ \rho}(b,c)
&=p^{-2}\sum_{x \in\mathbb{F}_{q}}
\sum_{y\in\mathbb{F}_{p}}\zeta_p^{y\mathrm{Tr}(ax^{p^k+1} )-cy}
\sum_{z\in\mathbb{F}_{p}}\zeta_p^{z\mathrm{Tr}(bx)-\rho z}\nonumber \\
&= \frac{n_c}{p}
+ p^{-2}\sum_{x  \in\mathbb{F}_{q} } \Bgk{1+
	\sum_{y\in\mathbb{F}_{p}^*}\zeta_p^{y\mathrm{Tr}(a x^{p^k+1} )-cy} }
\sum_{z\in\mathbb{F}_{p}^*}\zeta_p^{z\mathrm{Tr}(bx )-\rho z}\nonumber \\
&= \frac{n_c}{p}  + p^{-2} B(b,c),
\end{align}
where
\begin{align}\label{eq:B}
B(b,c) =\sum_{y\in\mathbb{F}_{p}^*}\zeta_p^{-cy}
\sum_{z\in\mathbb{F}_{p}^*}\zeta_p^{ -\rho z}
\sum_{x  \in\mathbb{F}_{q} }
\zeta_p^{\mathrm{Tr}(ay x^{p^k+1} +bzx)  }.
\end{align}

We are going to determine the values of $  B(b,c) $ in Lemmas \ref{lem:B1} and \ref{lem:B2}.
For later use, we set $ f_a(X)= a^{p^k} X^{p^{2k}} + a X \in \mathbb{F}_q[X]  $ for $  a \in\mathbb{F}_q^* $ in the sequel.
\begin{lemma}\label{lem:B1}
	If $  a^{(q-1)/(p^d+1)} \neq (-1)^s $, then $ f_a(X)=-b^{p^k} $ has a solution $ \gamma $ in $ \mathbb{F}_q $ and the following assertions hold.
	\begin{itemize}
		\item [1.] If $ c=0 $, then 	
		\begin{align*}
		B (b,0)=\left\{\begin{array}{lll}
		-(-1)^s (p-1)p^m
		&& \textup{ if }  \mathrm{Tr}(a \gamma^{p^k+1})=0,\\
		(-1)^s  p^{m}
		&& \textup{ if }   \mathrm{Tr}(a \gamma^{p^k+1})\neq 0.
		\end{array}
		\right.
		\end{align*}
		\item [2.] If $ c\neq 0 $, then 	
		\begin{align*}
		B (b,c)=\left\{\begin{array}{lll}
		(-1)^s p^m
		&& \textup{ if }  \mathrm{Tr}(a \gamma^{p^k+1})=0,\\
		(-1)^s \eta(\rho^2-4c\mathrm{Tr}(a \gamma^{p^k+1}) ) p^{m+1}+(-1)^s  p^{m}
		&& \textup{ if }   \mathrm{Tr}(a \gamma^{p^k+1}) \neq 0.
		\end{array}
		\right.
		\end{align*}
	\end{itemize}
	
\end{lemma}

\begin{proof}
	If $  a^{(q-1)/(p^d+1)} \neq (-1)^s $, by Lemma \ref{lem:fun}, we know the equation
	$ f_a(X) =a^{p^k}X^{p^{2k}} + a X $ is a permutation polynomial over $ \mathbb{F}_q $ and
	$ f_a(X) =-b^{p^k}  $ has a unique solution $ \gamma $ in
	$ \mathbb{F}_{q} $. Thus $ y^{-1}z\gamma $ is the unique solution for
	$ f_{ay}(X) =-{(bz)}^{p^k}  $ for any $ y,z \in \mathbb{F}_p^* $.
	According to Lemma \ref{lem:S(ab)}, $ S(ay,bz) = (-1)^s p^m \chi_1(-ay (y^{-1}z\gamma)^{p^k+1} ) $.
	It follows from \eqref{eq:B} that
	\begin{align*}
	B(b,c) &=\sum_{y\in\mathbb{F}_{p}^*}\zeta_p^{-cy}
	\sum_{z\in\mathbb{F}_{p}^*}\zeta_p^{ -\rho z}
	S(ay,bz)\\
	&= (-1)^s p^m \sum_{y\in\mathbb{F}_{p}^*}\zeta_p^{-cy}
	\sum_{z\in\mathbb{F}_{p}^*}\zeta_p^{-\frac{z^2}{y} \mathrm{Tr}(a \gamma^{p^k+1})  -\rho z}.
	\end{align*}
	
	If $ \mathrm{Tr}(a \gamma^{p^k+1}) =0 $, then
	\begin{align*}
	B(b,c) = (-1)^s p^m \sum_{y\in\mathbb{F}_{p}^*}\zeta_p^{-cy}
	\sum_{z\in\mathbb{F}_{p}^*}\zeta_p^{   -\rho z} ,
	\end{align*}
	leading to the desired results.
	
	Now we assume that $ \mathrm{Tr}(a \gamma^{p^k+1}) \neq 0 $ in the rest of the proof. It follows that
	\begin{align*}
	B(b,0) & = (-1)^s p^m \sum_{y\in\mathbb{F}_{p}^*}
	\sum_{z\in\mathbb{F}_{p}^* }\zeta_p^{-\frac{z^2}{y} \mathrm{Tr}(a \gamma^{p^k+1})  -\rho z} \\
	& = (-1)^s p^m
	\sum_{z\in\mathbb{F}_{p}^* }\zeta_p^{-\rho z}
	\sum_{y\in\mathbb{F}_{p}^*}   \zeta_p^{y  } =(-1)^s p^m.
	\end{align*}
	If $ c \neq 0 $, we have from Lemma \ref{lm:expo sum} that
	\begin{align*}
	B (b,c) & =(-1)^s p^m \sum_{y\in\mathbb{F}_{p}^*}\zeta_p^{-cy}
	\Bgk{
		\sum_{z\in\mathbb{F}_{p}}\zeta_p^{-\frac{z^2}{y} \mathrm{Tr}(a \gamma^{p^k+1})  -\rho z}-1}\\
	&= (-1)^s p^m \sum_{y\in\mathbb{F}_{p}^*}\zeta_p^{-cy}
	\zeta_p^{\frac{\rho^2 y}{4  \mathrm{Tr}(a \gamma^{p^k+1})}  }
	\eta\Bgk{-\frac{   \mathrm{Tr}(a \gamma^{p^k+1})}{y}} G
	+ (-1)^s p^m \\
	&= (-1)^s p^m \eta(-   \mathrm{Tr}(a \gamma^{p^k+1})) G
	\sum_{y\in\mathbb{F}_{p}^*}\zeta_p^{\big(\frac{\rho^2 }{4  \mathrm{Tr}(a \gamma^{p^k+1})}-c \big) y}  \eta(y)
	+ (-1)^s p^m  \\	
	&= \left\{\begin{array}{lll}
	(-1)^s p^m  && \textup{ if }   \mathrm{Tr}(a \gamma^{p^k+1})=\rho^2/(4c)\\
	(-1)^s p^m \eta(4c  \mathrm{Tr}(a \gamma^{p^k+1})-\rho^2) G^2
	+ (-1)^s p^m
	&& \textup{ if }   \mathrm{Tr}(a \gamma^{p^k+1})\neq\rho^2/(4c)
	\end{array}
	\right. \\
	&= \left\{\begin{array}{lll}
	(-1)^s p^m  && \textup{ if }   \mathrm{Tr}(a \gamma^{p^k+1})=\rho^2/(4c),\\
	(-1)^s p^{m+1} \eta(\rho^2-4c  \mathrm{Tr}(a \gamma^{p^k+1}))
	+ (-1)^s p^m
	&& \textup{ if }   \mathrm{Tr}(a \gamma^{p^k+1})\neq\rho^2/(4c)  .
	\end{array}
	\right.
	\end{align*}
	This gives the desired assertion, completing the whole proof.   
\end{proof}

\begin{lemma}\label{lem:B2}
	Let $  a^{(q-1)/(p^d+1)} = (-1)^s $. If $ f_a(X)=-b^{p^k} $ has no solution in $ \mathbb{F}_q $, then $ B(b,c)=0 $ for $ c \in \mathbb{F}_p $. Suppose that
	$ f_a(X)=-b^{p^k} $ has a solution $ \gamma $ in $ \mathbb{F}_q $, we have
	\begin{itemize}
		\item [1.] If $ c=0 $, then 	
		\begin{align*}
		B (b,0)=\left\{\begin{array}{lll}
		(-1)^s (p-1)p^{m+d}
		&& \textup{ if }  \mathrm{Tr}(a \gamma^{p^k+1})=0,\\
		-(-1)^s  p^{m+d}
		&& \textup{ if }   \mathrm{Tr}(a \gamma^{p^k+1})\neq 0.
		\end{array}
		\right.
		\end{align*}
		\item [2.] If $ c\neq 0 $, then 	
		\begin{align*}
		B (b,c)=\left\{\begin{array}{lll}
		-(-1)^s p^{m+d}
		&& \textup{ if }  \mathrm{Tr}(a \gamma^{p^k+1})=0,\\
		-(-1)^s \eta(\rho^2-4c\mathrm{Tr}(a \gamma^{p^k+1}) ) p^{m+d+1}-(-1)^s  p^{m+d}
		&& \textup{ if }   \mathrm{Tr}(a \gamma^{p^k+1}) \neq 0.
		\end{array}
		\right.
		\end{align*}
	\end{itemize}
	
\end{lemma}
\begin{proof}
	Let $  a^{(q-1)/(p^d+1)} = (-1)^s $. By \eqref{eq:B},
	\begin{align*}
	B_a(b,\rho) =\sum_{y\in\mathbb{F}_{p}^*}\zeta_p^{-cy}
	\sum_{z\in\mathbb{F}_{p}^*}\zeta_p^{ -\rho z}
	S(ay,bz).
	\end{align*}
	For $ y,z\in\mathbb{F}_{p}^* $, it follows from Lemma \ref{lem:S(ab)} that $ S(ay,bz)=0 $ unless the equation $ f_{a }(X)=-b^{p^k} $ is solvable.
	Note that $ f_a(X)=a  X^{p^{2k}} + a X $ is not a permutation polynomial over $ \mathbb{F}_q $ by Lemma \ref{lem:fun}.
	By a similar argument as above, we obtain the desired conclusions. The details are omitted.
\end{proof}

\begin{lemma}\label{lem:S}
	With the notation introduced above, we denote
	\begin{align*}
	S = \{b \in \mathbb{F}_q: f_a(X)=-b^{p^k} \textup{ is solvable in } \mathbb{F}_q \}.
	\end{align*}
	If $  a^{(q-1)/(p^d+1)} = (-1)^s $, then $ \#S=p^{e-2d} $.
\end{lemma}
\begin{proof}
	Taking into account that $ e/d $ is even and $  a^{(q-1)/(p^d+1)} = (-1)^s $, we get from Lemma \ref{lem:fun} that $ f_a(X)=-b^{p^k} $ has $ p^{2d} $ solutions in $ \mathbb{F}_q $, where $ b \in S $. For two distinct elements $ b_1 $ and $ b_2 $ in $ \mathbb{F}_q $, there are no common solutions for equations $ f_a(X)=-b_1^{p^k} $ and $ f_a(X)=-b_2^{p^k} $. On the other hand, for each $ x \in  \mathbb{F}_q  $, we know $ f_a(x)  $ is in $ \mathbb{F}_q $ and there exits an element $ b $ such that $ f_a(x)=-b^{p^k}  $. Hence we must have $ \#S \cdot p^{2d}=p^e $ giving the desired conclusion.
\end{proof}

With the above preparations, we can prove our main results listed in Section \ref{sec:main results}. There are four cases to consider:
\begin{itemize}
	\item [(1)] $ c=0 $ and $ a^{(q-1)/(p^d+1)} \neq (-1)^s  $,
	
	\item [(2)]  $ c=0 $ and $ a^{(q-1)/(p^d+1)} = (-1)^s  $,
	
	\item [(3)]  $ c\neq 0 $ and $ a^{(q-1)/(p^d+1)} \neq (-1)^s  $,
	
	\item [(4)]  $ c \neq 0 $ and $ a^{(q-1)/(p^d+1)} = (-1)^s  $.
\end{itemize}

\subsubsection{The first case where $ c=0 $ and $ a^{(q-1)/(p^d+1)} \neq (-1)^s  $}\label{case1}

In this subsection, we assume that $ c=0 $ and $ a^{(q-1)/(p^d+1)} \neq (-1)^s  $. Recall that $ s=m/d $ and $ \gamma $ is the unique solution for the equation $ f_a(X) =-b^{p^k}  $. By \eqref{eq:N0}, \eqref{eq:Nab}, Lemmas \ref{lem:length} and \ref{lem:B1}, we have the following two lemmas.
\begin{lemma}\label{lem:Na00}
	If $ a^{(q-1)/(p^d+1)} \neq (-1)^s  $, $ b\in \mathbb{F}_{q}^* $ and $ \rho \in \mathbb{F}_{p}^* $, then
	\begin{align*}
	N_{\rho}(b,0)
	&= \left\{\begin{array}{lll}
	p^{e-2}   && \textup{ if }  \mathrm{Tr}(a \gamma^{p^k+1})=0,\\
	p^{e-2} +(-1)^s p^{m-1}  && \textup{ if } \mathrm{Tr}(a \gamma^{p^k+1})\neq 0.
	\end{array}
	\right.
	\end{align*}
\end{lemma}
\begin{lemma}\label{lem:Na0b}
	If $ a^{(q-1)/(p^d+1)} \neq (-1)^s  $ and $ b\in \mathbb{F}_{q}^* $, then we have
	\begin{align*}
	N_0(b,0)
	= \left\{\begin{array}{lll}
	p^{e-2}+(-1)^s(p-1) p^{m-1}    && \textup{ if }  \mathrm{Tr}(a \gamma^{p^k+1})=0,\\
	p^{e-2}  && \textup{ if } \mathrm{Tr}(a \gamma^{p^k+1}) \neq 0.
	\end{array}
	\right.
	\end{align*}
\end{lemma}

Now we are in a position to prove Theorem \ref{thm:D0}.
\begin{proof}[Proof of Theorem \ref{thm:D0}]
	Denote
	\begin{align*}
	w_1&= (p-1) p^{e-2},\\
	w_2&= (p-1) (p^{e-2}+(-1)^s p^{m-1}).
	\end{align*}
	
	The code $ C_{D_0} $ has length $ n_0-1=
	p^{e-1} +	(-1)^s (p-1)p^{m-1} -1 $ and dimension $ e $, since $ wt(\mathsf{c}_b)  >0 $ for each $ b \in \mathbb{F}_{q}^* $. By the first two Pless Power Moments (see \cite{VPless2003funda}, page 260) the frequency $ A_{w_i} $ of $ w_i $ satisfies the following equations
	\begin{align*}
	\left\{\begin{array}{lll}
	A_{w_1}+A_{w_2}&=p^e-1,\\
	w_1A_{w_1}+w_2 A_{w_2}&=p^{e-1}(p-1)(n_0-1).
	\end{array}
	\right.
	\end{align*}
	Solving the equations gives that
	\begin{align*}
	\left\{\begin{array}{lll}
	A_{w_1}&=p^{e-1}+(-1)^s (p-1)p^{m-1}-1,\\
	 A_{w_2}&=(p-1)(p^{e-1}-(-1)^s p^{m-1}) .
	\end{array}
	\right.
	\end{align*}
	This leads to the weight enumerator and complete weight enumerator given in Theorem \ref{thm:D0}.
\end{proof}

\subsubsection{The second case where $ c=0 $ and $ a^{(q-1)/(p^d+1)} = (-1)^s  $}\label{case2}

In this subsection, we assume that $ c=0 $ and $ a^{(q-1)/(p^d+1)} = (-1)^s  $. By \eqref{eq:N0}, \eqref{eq:Nab}, Lemmas \ref{lem:length} and \ref{lem:B2}, we have the following two lemmas.
\begin{lemma}\label{lem:Na01}
	Let $ a^{(q-1)/(p^d+1)} = (-1)^s  $, $ b\in \mathbb{F}_{q}^* $ and $ \rho \in \mathbb{F}_{p}^* $.
	If $ f_a(X)=-b^{p^k} $ has no solution in $ \mathbb{F}_q $, then
	\begin{align*}
	N_{\rho}(b,0) =	p^{e-2} - (-1)^s (p-1)p^{m+d-2}.
	\end{align*}
	If $ f_a(X)=-b^{p^k} $ has a solution $ \gamma $ in $ \mathbb{F}_q $, then
	\begin{align*}
	N_{\rho}(b,0)
	&= \left\{\begin{array}{lll}
	p^{e-2}   && \textup{ if }  \mathrm{Tr}(a \gamma^{p^k+1})=0,\\
	p^{e-2} -(-1)^s p^{m+d-1}  && \textup{ if } \mathrm{Tr}(a \gamma^{p^k+1})\neq 0.
	\end{array}
	\right.
	\end{align*}
\end{lemma}
\begin{lemma}\label{lem:Na02}
	Let $ a^{(q-1)/(p^d+1)} = (-1)^s  $ and $ b\in \mathbb{F}_{q}^* $. If $ f_a(X)=-b^{p^k} $ has no solution in $ \mathbb{F}_q $, then
	\begin{align*}
	N_0(b,0) =	p^{e-2} - (-1)^s(p-1) p^{m+d-2}.
	\end{align*}
	If $ f_a(X)=-b^{p^k} $ has a solution $ \gamma $ in $ \mathbb{F}_q $, then we have
	\begin{align*}
	N_0(b,0)
	= \left\{\begin{array}{lll}
	p^{e-2}-(-1)^s(p-1) p^{m+d-1}    && \textup{ if }  \mathrm{Tr}(a \gamma^{p^k+1})=0,\\
	p^{e-2}  && \textup{ if } \mathrm{Tr}(a \gamma^{p^k+1}) \neq 0.
	\end{array}
	\right.
	\end{align*}
\end{lemma}

Now it comes to prove Theorem \ref{thm:D1}.
\begin{proof}[Proof of Theorem \ref{thm:D1}]
	Denote
	\begin{align*}
	w_1&=(p-1)(p^{e-2}-(-1)^s(p-1)p^{m+d-2}),\\
	w_2&= (p-1) p^{e-2},\\
	w_3&= (p-1) (p^{e-2}-(-1)^s p^{m+d-1}).
	\end{align*}
	
	The code $ C_{D_0} $ has length $ n_0-1=
	p^{e-1} -	(-1)^s(p-1) p^{m+d-1} -1 $ and dimension $ e $.
	It follows from Lemma \ref{lem:S} that $ A_{w_1}=p^e-p^{e-2d} $.
	By the first two Pless Power Moments (see \cite{VPless2003funda}, page 260) the frequency $ A_{w_i} $ of $ w_i $ satisfies the following equations
	\begin{align*}
	\left\{\begin{array}{ll}
	A_{w_1}+A_{w_2}+A_{w_3}&=p^e-1,\\
	w_1A_{w_1}+w_2 A_{w_2}+w_3 A_{w_3}&=p^{e-1}(p-1)(n_0-1).
	\end{array}
	\right.
	\end{align*}
	Solving the equations gives that
	\begin{align*}
	\left\{\begin{array}{ll}
	A_{w_2}&=p^{e-2d-1}-(-1)^s(p-1) p^{m-d-1}-1,\\
	 A_{w_3}&=(p-1)(p^{e-2d-1}+(-1)^s p^{m-d-1}) .
	\end{array}
	\right.
	\end{align*}
	This leads to the weight enumerator and complete weight enumerator given in Theorem \ref{thm:D1}.
\end{proof}

\subsubsection{The third case where $ c\neq 0 $ and $ a^{(q-1)/(p^d+1)} \neq (-1)^s  $}\label{case3}

In this subsection, we assume that $ c\neq 0 $ and $ a^{(q-1)/(p^d+1)} \neq (-1)^s  $. Recall that $ \gamma $ is the unique solution for the equation $ f_a(X) =-b^{p^k}  $. By \eqref{eq:Nab} and Lemma \ref{lem:B1}, we have the values of $ N_{\rho}(b,c) $.
\begin{lemma}\label{lem:Na10}
	Assume that $ a^{(q-1)/(p^d+1)} \neq (-1)^s  $, $ b\in \mathbb{F}_{q}^* $, $ c \in \mathbb{F}_{p}^* $ and $ \rho \in \mathbb{F}_{p}^* $. Then
	\begin{align*}
	N_{\rho}(b,c)
	&= \left\{\begin{array}{lll}
	p^{e-2}   && \textup{ if }  \mathrm{Tr}(a \gamma^{p^k+1})=0,\\
	p^{e-2}+(-1)^s \eta(\rho^2-4c\mathrm{Tr}(a \gamma^{p^k+1})) p^{m-1}    && \textup{ if } \mathrm{Tr}(a \gamma^{p^k+1}) \neq 0.
	\end{array}
	\right.
	\end{align*}
\end{lemma}
\begin{lemma}\label{lem:Na11}
	Assume that $ a^{(q-1)/(p^d+1)} \neq (-1)^s  $, $ b\in \mathbb{F}_{q}^* $ and  $ c \in \mathbb{F}_{p}^* $. Then we have
	\begin{align*}
	N_0(b,c)
	= \left\{\begin{array}{lll}
	p^{e-2}-(-1)^s p^{m-1}    && \textup{ if }  \mathrm{Tr}(a \gamma^{p^k+1})=0,\\
	p^{e-2} +(-1)^s \eta(-c\mathrm{Tr}(a \gamma^{p^k+1}))p^{m-1}  && \textup{ if } \mathrm{Tr}(a \gamma^{p^k+1}) \neq 0.
	\end{array}
	\right.
	\end{align*}
\end{lemma}
\begin{proof}
	It follows from \eqref{eq:N0}, Lemmas \ref{lem:length} and \ref{lem:Na10} that
	\begin{align*}
	N_0(b,c)&=n_c-\sum_{\rho \neq 0 }N_{\rho}(b,c)\\
	&	= \left\{\begin{array}{lll}
	n_c-(p-1)p^{e-2}   && \textup{ if }  \mathrm{Tr}(a \gamma^{p^k+1})=0,\\
	n_c- (-1)^s p^{m-1} \sum_{\rho \neq 0 } \eta(\rho^2-4c\mathrm{Tr}(a \gamma^{p^k+1}))   && \textup{ if } \mathrm{Tr}(a \gamma^{p^k+1}) \neq 0,
	\end{array}
	\right.
	\end{align*}
	where $ n_c=	p^{e-1} -(-1)^s p^{m-1} $.
	Taking into account that $ \sum_{\rho \neq 0 } \eta(\rho^2-4c\mathrm{Tr}(a \gamma^{p^k+1})) =-1-\eta(-c\mathrm{Tr}(a \gamma^{p^k+1})) $ by Lemma \ref{lm:eta(f)}, we get the desired results.
\end{proof}

The proof of Theorem \ref{thm:D2} is given below.
\begin{proof}[Proof of Theorem \ref{thm:D2}]
	In this case, it follows from Lemmas \ref{lem:Na10} and \ref{lem:Na11} that the weight of $ \mathsf{c}_b $ has possible values
	$ w_1=(p-1)p^{e-2} $ and $ w_2= (p-1)p^{e-2}-2(-1)^s p^{m-1}$. By a similar argument as above, we get the desired conclusions. The details are omitted here.
\end{proof}

\subsubsection{The fourth case where $ c\neq 0 $ and $ a^{(q-1)/(p^d+1)} = (-1)^s  $}\label{case4}

In this subsection, we assume that $ c\neq 0 $ and $ a^{(q-1)/(p^d+1)} = (-1)^s  $. By \eqref{eq:N0}, \eqref{eq:Nab}, Lemmas \ref{lm:eta(f)} and \ref{lem:B1}, we have the values of $ N_{\rho}(b,c) $ and $ N_0(b,c) $.
\begin{lemma}\label{lem:Na20}
	Let $ a^{(q-1)/(p^d+1)} = (-1)^s  $, $ b\in \mathbb{F}_{q}^* $, $ c \in \mathbb{F}_{p}^* $ and $ \rho \in \mathbb{F}_{p}^* $. If $ f_a(X)=-b^{p^k} $ has no solution in $ \mathbb{F}_q $, then
	\begin{align*}
	N_{\rho}(b,c) =	p^{e-2} + (-1)^s p^{m+d-2}.
	\end{align*}
	If $ f_a(X)=-b^{p^k} $ has a solution $ \gamma $ in $ \mathbb{F}_q $, then
	\begin{align*}
	N_{\rho}(b,c)
	&= \left\{\begin{array}{lll}
	p^{e-2}   && \textup{ if }  \mathrm{Tr}(a \gamma^{p^k+1})=0,\\
	p^{e-2} -(-1)^s \eta(\rho^2-4c\mathrm{Tr}(a \gamma^{p^k+1})) p^{m+d-1}   && \textup{ if } \mathrm{Tr}(a \gamma^{p^k+1})\neq 0.
	\end{array}
	\right.
	\end{align*}
\end{lemma}
\begin{lemma}\label{lem:Na21}
	Let $ a^{(q-1)/(p^d+1)} = (-1)^s  $, $ b\in \mathbb{F}_{q}^* $ and  $ c \in \mathbb{F}_{p}^* $. If $ f_a(X)=-b^{p^k} $ has no solution in $ \mathbb{F}_q $, then
	\begin{align*}
	N_0(b,c) =	p^{e-2} + (-1)^s  p^{m+d-2}.
	\end{align*}
	If $ f_a(X)=-b^{p^k} $ has a solution $ \gamma $ in $ \mathbb{F}_q $, then we have
	\begin{align*}
	N_0(b,c)
	= \left\{\begin{array}{lll}
	p^{e-2}+(-1)^s  p^{m+d-1}    && \textup{ if }  \mathrm{Tr}(a \gamma^{p^k+1})=0,\\
	p^{e-2}-(-1)^s\eta(-c\mathrm{Tr}(a \gamma^{p^k+1}))p^{m+d-1}  && \textup{ if } \mathrm{Tr}(a \gamma^{p^k+1}) \neq 0.
	\end{array}
	\right.
	\end{align*}
\end{lemma}

Now we begin to prove Theorem \ref{thm:D3}.
\begin{proof}[Proof of Theorem \ref{thm:D3}]
	In this case, it follows from Lemmas \ref{lem:Na20} and \ref{lem:Na21} that $ wt(\mathsf{c}_b) $ takes values in $ \{w_1=(p-1)(p^{e-2}+(-1)^s p^{m+d-2}), w_2=(p-1)p^{e-2} , w_3= (p-1)p^{e-2}+2(-1)^s p^{m+d-1}\} $. By a similar argument as above, we get the desired conclusions. The details are omitted here.
\end{proof}

\section{Concluding remarks}\label{sec:conclusion}

In this paper, we employed exponential sums to present the complete weight enumerator and weight enumerator
of $ C_{D_c}  $ with defining set $ D_c $. As introduced in \cite{YD2006}, any linear code over $ \mathbb{F}_p $ can be employed to construct secret sharing schemes with interesting access structures provided that
\begin{align*}
\frac{w_{min}}{w_{max}}>\frac{p-1}{p},
\end{align*}
where $w_{min}$ and $w_{max}$ denote the minimum and maximum nonzero weights in $C_{D}$, respectively. For the linear codes in Theorems \ref{thm:D0} and \ref{thm:D2}, we have $  \frac{w_{min}}{w_{max}}> \frac{p-1}{p} $ if $ m>2 $. Similarly for the linear codes in Theorems \ref{thm:D1} and \ref{thm:D3}, we have $  \frac{w_{min}}{w_{max}}> \frac{p-1}{p} $ if $ m>d+2 $.
We remark that the dimension of the code of this paper is small
compared with its length and this makes it suitable for the
application in secret sharing schemes with interesting access structures.

%

\begin{thebibliography}{10}
	\providecommand{\url}[1]{{#1}}
	\providecommand{\urlprefix}{URL }
	\expandafter\ifx\csname urlstyle\endcsname\relax
	\providecommand{\doi}[1]{DOI~\discretionary{}{}{}#1}\else
	\providecommand{\doi}{DOI~\discretionary{}{}{}\begingroup
		\urlstyle{rm}\Url}\fi
	
	\bibitem{AhnKaLi2016completegenelize}
	Ahn, J., Ka, D., Li, C.: Complete weight enumerators of a class of linear
	codes.
	\newblock Designs, Codes and Cryptography \textbf{83}, 83--99 (2017).
	
	\bibitem{Bae2015complete}
	Bae, S., Li, C., Yue, Q.: On the complete weight enumerators of some reducible
	cyclic codes.
	\newblock Discrete Mathematics \textbf{338}(12), 2275--2287 (2015).
	
	\bibitem{Blake1991}
	Blake, I.F., Kith, K.: {On the complete weight enumerator of Reed-Solomon
		codes}.
	\newblock SIAM J. Discret. Math. \textbf{4}(2), 164--171 (1991).
	
	\bibitem{coulter1998explicit}
	Coulter, R.S.: Explicit evaluations of some {W}eil sums.
	\newblock Acta Arithmetica \textbf{83}(3), 241--251 (1998).
	
	\bibitem{Coulter2002number}
	Coulter, R.S.: The number of rational points of a class of Artin-Schreier
	curves.
	\newblock Finite Fields and Their Applications \textbf{8}, 397--413 (2002).
	
	\bibitem{ding2015twodesign}
	Ding, C.: Linear codes from some 2-designs.
	\newblock IEEE Transactions on Information Theory \textbf{61}(6), 3265--3275
	(2015).
	
	\bibitem{ding2007generic}
	Ding, C., Helleseth, T., Kl\o{}ve, T., Wang, X.: {A generic construction of
		Cartesian authentication codes}.
	\newblock IEEE Transactions on Information Theory \textbf{53}(6), 2229--2235
	(2007).
	
	\bibitem{Ding2005auth}
	Ding, C., Wang, X.: A coding theory construction of new systematic
	authentication codes.
	\newblock Theoretical Computer Science \textbf{330}, 81--99 (2005).
	
	\bibitem{dingkelan2014binary}
	Ding, K., Ding, C.: Binary linear codes with three weights.
	\newblock IEEE Communications Letters \textbf{18}(11), 1879--1882 (2014).
	
	\bibitem{ding2015twothree}
	Ding, K., Ding, C.: A class of two-weight and three-weight codes and their
	applications in secret sharing.
	\newblock IEEE Transactions on Information Theory \textbf{61}(11), 5835--5842
	(2015).
	
	\bibitem{dinh2015recent}
	Dinh, H., Li, C., Yue, Q.: Recent progress on weight distributions of cyclic
	codes over finite fields.
	\newblock Journal of Algebra Combinatorics Discrete Structures and Applications
	\textbf{2}(1), 39--63 (2015).
	
	\bibitem{feng2008weight}
	Feng, K., Luo, J.: Weight distribution of some reducible cyclic codes.
	\newblock Finite Fields and Their Applications \textbf{14}, 390--409 (2008).
	
	\bibitem{helleseth2006}
	Helleseth, T., Kholosha, A.: Monomial and quadratic bent functions over the
	finite fields of odd characteristic.
	\newblock IEEE Transactions on Information Theory \textbf{52}(5), 2018--2032
	(2006).
	
	\bibitem{Heng2016cwecylic}
	Heng, Z., Yue, Q.: Complete weight distributions of two classes of cyclic
	codes.
	\newblock Cryptography and Communications \textbf{9}(3), 323--343 (2017).
	
	\bibitem{VPless2003funda}
	Huffman, W.C., Pless, V.: {Fundamentals of Error-Correcting Codes}.
	\newblock Cambridge University Press, Cambridge (2003).
	
	\bibitem{kith1989complete}
	Kith, K.: {Complete weight enumeration of Reed-Solomon codes}.
	\newblock Master's thesis, Department of Electrical and Computing Engineering,
	University of Waterloo  (1989).
	
	\bibitem{kuzmin1999complete}
	Kuzmin, A., Nechaev, A.: {Complete weight enumerators of generalized Kerdock
		code and linear recursive codes over Galois ring}.
	\newblock In: Workshop on Coding and Cryptography, pp. 332--336 (1999).
	
	\bibitem{LiYang2015cwe}
	Li, C., Bae, S., Ahn, J., Yang, S., Yao, Z.: Complete weight enumerators of
	some linear codes and their applications.
	\newblock Designs, Codes and Cryptography \textbf{81}, 153--168 (2016).
	
	\bibitem{li2015complete}
	Li, C., Yue, Q., Fu, F.W.: Complete weight enumerators of some cyclic codes.
	\newblock Designs, Codes and Cryptography \textbf{80}, 295--315 (2016).
	
	\bibitem{lidl1983finite}
	Lidl, R., Niederreiter, H.: Finite fields.
	\newblock Encyclopedia of Mathematics and its Applications. Reading,
	Massachusetts, USA: Addison-Wesley \textbf{20} (1983).
	
	\bibitem{macwilliams1977theory}
	MacWilliams, F.J., Sloane, N.J.A.: The Theory of Error-Correcting Codes,
	vol.~16.
	\newblock North-Holland Publishing, Amsterdam (1977).
	
	\bibitem{vega2012weight}
	Vega, G.: The weight distribution of an extended class of reducible cyclic
	codes.
	\newblock IEEE Transactions on Information Theory \textbf{58}(7), 4862--4869
	(2012).
	
	\bibitem{WangQiuyan2015complet}
	Wang, Q., Li, F., Ding, K., Lin, D.: Complete weight enumerators of two classes
	of linear codes.
	\newblock Discrete Mathematics \textbf{340}, 467--480 (2017).
	
	\bibitem{YangYao2017}
	 Yang, S., Yao, Z., Zhao, C.: The weight distributions of two classes of $ p $-ary cyclic codes with few weights.
	 \newblock Finite Fields and Their Applications \textbf{44}, 76--91 (2017).
	
	\bibitem{Yang2017constru}
	Yang, S., Kong, X., Tang, C.: A construction of linear codes and their complete
	weight enumerators.
	\newblock Finite Fields and Their Applications \textbf{48}, 196--226 (2017).
	
	\bibitem{Yang2016complete}
	Yang, S., Yao, Z.: Complete weight enumerators of a class of linear codes.
	\newblock Discrete Mathematics \textbf{340}, 729--739 (2017).
	
	\bibitem{YD2006}
	Yuan, J., Ding, C.: Secret sharing schemes from three classes of linear codes.
	\newblock IEEE Transactions on Information Theory \textbf{52}(1), 206--212
	(2006).
	
	\bibitem{zheng2015weightseveral}
	Zheng, D., Wang, X., Yu, L., Liu, H.: The weight enumerators of several classes
	of $p$-ary cyclic codes.
	\newblock Discrete Mathematics \textbf{338}, 1264--1276 (2015).
	
	\bibitem{Zhou2013fiveweight}
	Zhou, Z., Ding, C., Luo, J., Zhang, A.: A family of five-weight cyclic codes
	and their weight enumerators.
	\newblock IEEE Transactions on Information Theory \textbf{59}(10), 6674--6682
	(2013).
	
\end{thebibliography}

\end{document}